\documentclass[11pt]{article}

\newcommand{\etal}{et al.\ }

\usepackage{amsthm, amsmath, amssymb}
\usepackage{graphicx}
\usepackage{algorithm}
\usepackage{algorithmic}
\usepackage{xspace}
\usepackage{enumitem}
\usepackage{tabularx}
\usepackage{xcolor}
\usepackage{graphicx} 
\usepackage{tikz}
\usepackage{wrapfig}

\DeclareMathOperator{\E}{\mathbb E}

\newcommand{\R}{\mathbb{R}}

\newtheorem{theorem}{Theorem}[section]
\newtheorem{lemma}[theorem]{Lemma}

\newtheorem{definition}[theorem]{Definition}
\newtheorem{claim}[theorem]{Claim}

\newcommand{\set}[1]{\left\{#1\right\}}

\newcommand{\floor}[1]{\left\lfloor#1\right\rfloor}
\newcommand{\ceil}[1]{\left\lceil#1\right\rceil}

\newcommand{\calA}{{\mathcal A}}

\newcommand{\calR}{{\mathcal R}}

\newcommand{\eps}{\epsilon}
\newcommand{\bfTheta}{{\mathbf\Theta}}

\newcommand{\idle}{{\mathsf{idle}}}
\newcommand{\emp}{{\mathsf{empty}}}
\newcommand{\sfd}{{\mathsf{d}}}
\newcommand{\sfb}{{\mathsf{b}}}

\newcommand{\lpinterval}{\mbox{$\mathsf{LP}_\mathsf{interval}$}\xspace}
\newcommand{\lpchain}{\mbox{$\mathsf{LP}_\mathsf{chain}$}\xspace}
\newcommand{\lpdual}{\mbox{$\mathsf{LP}_\mathsf{chain-d}$}\xspace}

\begin{document}

\sloppy

\title{Better Unrelated Machine Scheduling for Weighted Completion Time via Random Offsets from Non-Uniform Distributions}

\author{Sungjin Im\thanks{ Electrical Engineering and Computer Science, University of California, 5200 N. Lake Road, Merced CA 95344. {\tt sim3@ucmerced.edu}. Supported in part by NSF grant CCF-1409130.}  \and Shi Li \thanks{Department of Computer Science and Engineering, University at Buffalo, 1 White Road, Buffalo, NY 14260.  {\tt shil@buffalo.edu}. Supported in part by NSF grant CCF-1566356.}} 

\date{}
\maketitle
\thispagestyle{empty}

\begin{abstract}

In this paper we consider the classic scheduling problem of minimizing total weighted completion time on unrelated machines when jobs have release times, i.e, $R | r_{ij} | \sum_j w_j C_j$ using  the three-field notation.  For this problem, a 2-approximation is known based on a novel convex programming  (J. ACM 2001 by Skutella). It has been a long standing open problem if one can improve upon this 2-approximation (Open Problem 8 in J. of Sched. 1999 by Schuurman and Woeginger). We answer this question in the affirmative by giving a 1.8786-approximation. We achieve this via a surprisingly simple linear programming, but a novel rounding algorithm and analysis. A key ingredient of our algorithm is the use of random offsets sampled from non-uniform distributions.

We also consider the preemptive version of the problem, i.e, $R | r_{ij},pmtn | \sum_j w_j C_j$. We again use the idea of sampling offsets from non-uniform distributions to give the first better than 2-approximation for this problem. This improvement also requires use of a configuration LP with variables for each job's complete schedules along with more careful analysis. For both non-preemptive and preemptive versions, we break the approximation barrier of 2 for the first time. 
\end{abstract}

\clearpage
\setcounter{page}{1}

	\section{Introduction}

Modern computing facilities serve a large number of jobs with different characteristics. To cope with this challenge, they are equipped with increasingly heterogeneous machines that are clustered and connected in networks, so that 
each job can be scheduled on a more suitable machine \cite{sahoo2004failure,mukherjee2009spatio,beloglazov2010energy}. Further, the large number of machines of different generations are deployed over a long period of time, increasing the heterogeneity. 
The scheduling decision must factor in the heterogeneity and communication overhead.

Unrelated machine scheduling is a widely studied classic model that captures various scenarios including the above. There is a set $J$ of jobs to be scheduled on a set $M$ of unrelated machines. Each job $j \in J$ can have an arbitrary processing time/size $p_{i,j}$ depending on the machine $i$ it gets processed; if $p_{i,j} = \infty$, then job $j$ cannot be scheduled on machine $i$. Furthermore, due to the communication delay, job $j$ is available for service only from time $r_{i,j}$, which can be also arbitrary depending on the job $j$ and the machine $i$ the job gets assigned to. The parameter $r_{i,j}$ is often called as job $j$'s arrival/release time.\footnote{For simplicity,  we will mostly assume that job $j$'s release time $r_{i,j}$ is the same for all machines. This will justify using a simpler notation $r_j$ in place of $r_{i,j}$. Like most of previous works, extending our result to release dates with dependency on machines is straightforward.}
Another parameter $w_j$ is used to capture job $j$'s importance. 

Minimizing total (weighted) completion time is one of the most  popular scheduling objectives that has been extensively studied, even dating back to 50's \cite{smith1956various}. 
The scheduler must assign each job $j$ to a machine $i$ and complete it. 
 We consider two settings, preemptive and non-preemptive schedules. In the non-preemptive setting, each job must be completed without interruption once it starts getting processed. On the other hand, in the preemptive setting, each job's processing  can be interrupted to process other jobs and be resumed later. In both cases, job $j$'s completion time is, if $j$ is assigned to machine $i$,  defined as the first time when the job gets processed for $p_{i,j}$ units of time. Then, the objective is to minimize $\sum_{j \in J} w_j C_j$. These two non-preemptive and preemptive versions can be described as $R|r_j|\sum_{j}w_jC_j$ and $R|r_j,pmtn|\sum_{j}w_jC_j$ respectively, using the popular three-field notation in scheduling literature. Both versions of the problem are strongly NP-hard even in the single machine setting \cite{lenstra1977complexity}, and are APX-hard even when all jobs are available for schedule at time 0 \cite{hoogeveen2001non}, in which case preemption does not help. 
 
For the non-preemptive case, Skutella gave a 2-approximation based on a novel convex programming \cite{Skutella01}, which improved upon the $(2+\eps)$-approximation based on linear programming \cite{schulz2002scheduling}. It has been an outstanding open problem if there exists a better than 2-approximation \cite{Skutella01,schuurman1999polynomial,schulz2002scheduling,kumar2008minimum,sviridenko2013approximating}. In particular, it is listed in \cite{schuurman1999polynomial} as one of the top 10 opens problems in the field of approximate scheduling algorithms; see the Open Problem 8. When jobs have no arrival times, i.e. $r_{i,j} = 0$ for all $i,j$, very recently Bansal \etal \cite{BansalSS15} gave a better than 1.5-approximation in a breakthrough result, improving upon the previous best 1.5-approximations due to Skutella \cite{Skutella01} and Sethuraman and Squillante \cite{sethuraman1999optimal}. In fact, the Open Problem 8 consists of two parts depending on whether jobs have release times or not. 
Bansal \etal \cite{BansalSS15} solved the first part of Open Problem 8, and the second part still remained open.

\subsection{Our Results}

In this paper, we answer the second part of the open problem in the affirmative by giving a better than 2-approximation. 

\begin{theorem}[Section~\ref{sec:non-preemptive}]
	\label{thm:non-preemptive}
	For a constant $\alpha < 1.8786$, there exists an $\alpha$-approximation for $R|r_j|\sum_{j}w_jC_j$.
\end{theorem}

Surprisingly, we give this result by rounding a very simple and natural LP that has not been studied in previous works.
Our LP can be viewed as a stronger version of the time-indexed LP in \cite{schulz2002scheduling}, by taking the non-preemption requirement into consideration.  However, even with this stronger LP, the rounding algorithm in \cite{schulz2002scheduling} does not yield a better than 2-approximation (see the discussion about use of uniform distribution in Section~\ref{sec:n-analysis}), and we believe this is why the previous works overlooked this simple LP.  Improving the $2$-approximation ratio requires not only the stronger LP, but also novel rounding algorithm and analysis. 

Our result also gives a positive answer to the conjecture made by Sviridenko and Wiese \cite{sviridenko2013approximating}. They considered a configuration LP where there is a variable for every machine $i \in M$ and  subset of jobs $S \subseteq J$. The variable is associated with the optimal total weighted completion time of the jobs in $S$ on machine $i$. They showed that one can solve their LP within a factor of $1+\eps$, but could not give a better than 2-approximation, conjecturing that their LP have an integrality gap strictly less than 2.  

Indeed, one can show that the configuration LP of \cite{sviridenko2013approximating} is the strongest among all convex programmings of the following form (see Appendix~\ref{sec:configuration-LP}): minimize $\sum_{i \in M} f_i(x_i)$ subject to $\sum_{i \in M}x_{i,j} = 1$ for every $j \in J$ and $x_{i, j} \geq 0$ for every $i \in M, j \in J$, where $x_i = (x_{i,j})_{j \in J} \in [0,1]^J$ and $f_i$ is some convex function over $[0, 1]^J$ such that if $x_i \in \{0, 1\}^J$, then $f_i(x_i)$ is at most the total weighted completion time of scheduling jobs $\{j:x_{i,j}=1\}$ optimally on machine $i$.   All results mentioned in this paper (including our results) are based on programmings of this form and thus the configuration LP is the strongest among them.  Hence, our result gives a $1.8786$ upper bound on the integrality gap of the configuration LP. 

With a solution to the configuration LP, one can derive a natural independent rounding algorithm. For each job $j$, independently assign $j$ to a machine $i$ with probability $x_{i,j}$. Then for every machine $i$, we schedule all jobs assigned to $i$; this can be done optimally if all release times are $0$ \cite{smith1956various}, and nearly optimally (within $(1+\eps)$ factor) in general \cite{afrati1999approximation,hall1997scheduling}.  When all jobs have release time $0$, the algorithm gives a $1.5$-approximation. However, \cite{BansalSS15} showed this independent rounding algorithm can not give a better than 1.5-approximation, which motivated them to develop a clever dependent rounding algorithm. 

For $R|r_j|\sum_{j}w_jC_j$, the independent rounding algorithm is known to give a $2$-approximation \cite{schulz2002scheduling,Skutella01}. In contrast to the status for $R||\sum_{j}w_jC_j$, no matching lower bound was known for this algorithm. Our result indirectly shows that the independent rounding can achieve $1.8786$-approximation.  Thus we do not need to apply the sophisticated dependence rounding scheme of \cite{BansalSS15}, which only led to a tiny improvement on the approximation ratio for $R| |\sum_{j}w_jC_j$.  We complement our positive result by showing that the independent rounding algorithm can not give an approximation ratio better than $e/(e-1) \approx 1.581$. 
\begin{theorem}[Section~\ref{sec:lb}]
	\label{thm:lb}
	There is an instance for which the independent rounding gives an approximation ratio worse than $e / (e - 1) - \eps \geq 1.581-\eps$ for any $\eps > 0$. 
\end{theorem}

\smallskip

We continue to study the preemptive case. In the preemptive case, two variants were considered in the literature depending on whether jobs can migrate across machines or must be completed scheduled on one of the machines. If migration is not allowed, the work in \cite{schulz2002scheduling} still gives a $(2+\eps)$-approximation since the LP therein is a relaxation for preemptive schedules but the rounding outputs a non-preemptive schedule. 
If migration is allowed, \cite{Skutella01} gives a 3-approximation. Our main result for the preemptive case is the first better than 2-approximation when migration is not allowed.

\begin{theorem}[Section~\ref{sec:preemptive}]
	\label{thm:preemptive}
	For a constant $\alpha < 1.99971$, there exists an $\alpha$-approximation for $R|r_j,pmtn|\sum_{j}w_jC_j$.
\end{theorem}

We note that our algorithm is based on a stronger linear programming relaxation. The configuration LP of \cite{sviridenko2013approximating} is for non-preemptive schedules hence not usable for preemptive schedules. Our LP is a different type of configuration LP where there are variables for each job's complete schedules. While we use an  LP  for 
preemptive schedules, we output a non-preemptive schedule.

\subsection{Our Techniques}
	
	As mentioned before, we give a better than 2-approximation for the non-preemptive case based on a very simple LP.  	In this LP, we have an indicator variable $y_{i,j,s}$ which is 1 if job $j$ starts at time $s$ on machine $i$. Then, we add an obvious constraint that no more than one job can be processed at any time on any machine. 
	This LP has a pseudo-polynomial size but can be reduced to a polynomial size using standard techniques with a loss of $(1+\eps)$ factor in approximation.
	
	As mentioned earlier, our algorithm falls into the independent rounding framework: we assign each job $j$ to machine $i$ with probability $x_{i,j}=\sum_{s} y_{i,j,s}$ independently following the optimal LP solution. Then, it remains to schedule jobs assigned to each machine.\footnote{Since $1|r_j|\sum_j w_jC_j$ admits a PTAS, given the set of jobs assigned to $i$, one can find a $(1+\eps)$-approximately optimal schedule on $i$. However, it is hard to directly relate this schedule to the fractional solution.}   Any solution to our LP is also a solution to the LP in \cite{schulz2002scheduling}.  When restricted to a solution to our LP,  the rounding algorithm of \cite{schulz2002scheduling} works as follows. For every $j$ that is assigned to $i$, we choose $s_j=s$ randomly with probability proportional to $y_{i, j, s}$. Then we choose $\tau_j$ uniformly at random from $[s_j, s_j+p_{i,j}]$; here $\tau_j - s_j$ can be viewed as a random extra offset applied to $j$. We schedule jobs assigned to $i$ non-preemptively in increasing order of $\tau_j$ values.  	While this gives a 2-approximation, this is the best one can obtain using their LP since it has a matching integrality gap.  Even with our stronger LP, the algorithm only gives a 2-approximation. 
	
	We use a more sophisticated distribution to sample $\tau_j$ for individual jobs. Discovering such a distribution and showing how it helps improve the approximation ratio requires a novel analysis. We are not the first that use non-uniform distributions for scheduling problems. Goemans \etal \cite{goemans2002single} used non-uniform distributions in their $\alpha$-point rounding for the single machine scheduling, i.e. $1|r_j|\sum_{j}w_jC_j$ to give a 1.6853-approximation. However, their analysis does not lend itself to multiple machines. The LP objective considered in \cite{goemans2002single} uses the notion of fractional completion time, which views a job $j$ of size $p_j$ as consisting of $p_{j}$ unit pieces with weight $w_j / p_j$. In this view, the optimal schedule trivially follows from the simple greedy Smith rule. \cite{goemans2002single} heavily uses this special structure to get a better than 2 approximation. However, this relaxation inherently loses a factor 2 when applied to  multiple machines even with some correction terms \cite{schulz2002scheduling,Skutella01}. Hence to overcome the 2-approximation barrier, one has to deviate from this relaxation and the special structure used in \cite{goemans2002single}, which calls for use of a stronger LP along with new algorithms and/or analysis. Intuitions on the effect of non-uniform distributions can be found in Section~\ref{sec:non-preemptive}, particularly in discussion of the limitations of uniform distributions.

	\smallskip
	As mentioned before, the preemptive result requires an even stronger LP where there is a variable for each job's complete schedule. Since preemption is allowed, even when all parameters are polynomially bounded, the LP has exponentially many variables. We solve  this LP by solving its dual with help of a separation oracle. While the algorithm for the non-preemptive case naturally extends to the preemptive case, the analysis doesn't. At a high level, the analysis for both cases needs to carefully handle the interaction between busy times and idle times which both can contribute jobs delays. Non-preemptive schedules possess better structural properties which allow us to break down the analysis into that for each time step. However, preemptive schedules lack such properties and require a different analysis of a somewhat amortized flavor.

\subsection{Other Related Work} 
The first non-trivial $O(\log^2 n)$-approximation for $R | r_{ij} | \sum_j w_j C_j$ was given by Stein \etal \cite{phillips1997task} using a hypergraph matching. Then, subsequent works \cite{hall1997scheduling,schulz2002scheduling,Skutella01}  gave constant approximations, culminating in a 2-approximation \cite{Skutella01} which was the best known prior to our work.  The work in \cite{hall1997scheduling} uses the celebrated rounding for the generalized assignment problem \cite{shmoys1993approximation} to round an LP with intervals of doubling lengths, thereby giving a 16/3-approximation. As mentioned before, \cite{schulz2002scheduling} gives a $(2+\eps)$-approximation, and there is an easy instance of matching integrality gap for their LP.  Subsequently, Skutella gave a 2-approximation using a convex programming \cite{Skutella01}, which is tight since the CP has an integrality gap of 2.  When machines are identical or uniformly related, a special case of unrelated machines, PTASes are known \cite{afrati1999approximation,skutella2000ptas,chekuri2001ptas}. 

Minimizing makespan or equivalently the maximum completion time is a closely related objective. For this problem when all jobs arrive at time 0, Lensta \etal gave a 2-approximation and showed it does not admit a better than 1.5 approximation unless P = NP \cite{lenstra1990approximation}. Reducing this gap remains open. Svensson showed that one can estimate the optimal makespan within a factor of $33/17 + \eps$ for the special case of restricted assignment \cite{svensson2012santa}. For other interesting special cases, see \cite{ebenlendr2008graph} and its follow-up works. For the dual objective of maximizing the minimum load on any machine, see~\cite{bansal2006santa,asadpour2010approximation,asadpour2008santa,chakrabarty2009allocating,feige2008allocations}. For the minimizing $\ell_p$ norms of completion times, see~\cite{azar2005convex,kumar2009unified}. 

For the objective of minimizing total flow time, i.e. $\sum_j (C_j - r_j)$, a poly-logarithmic approximation is known \cite{BansalK15}. For  earlier works for the restricted assignment case, see~\cite{GargK07,GargKM08}.  Due to the vast literature on scheduling, our discussion on related work is necessarily incomplete. For a nice survey and more pointers, see~\cite{chekuri2004approximation}.

	\section{Non-Preemptive Scheduling}
	\label{sec:non-preemptive}

We begin by giving an LP for the non-preemptive case. To present our algorithm and analysis more transparently, we assume that all parameters are polynomially bounded, i.e. all $w_j, r_j, p_{i,j}$ are $\texttt{poly}(|J|, |M|)$. Although we can also handle the case when $p_{ij} = \infty$ by not allowing $j$ to be scheduled on machine $i$, we assume such a case does not happen since the extension is straightforward. These simplifying assumptions will be removed in Section~\ref{sec:non-preemptive-a}. 

Define $T:= \sum_{i,j} p_{i,j} + \max_{j} r_j$ so that any `reasonable' scheduler can complete all jobs by the time $T$. Throughout this section, $s$ is always an integer. 
	\begin{equation}
		\min \qquad \sum_{i \in M,\ j \in J,\ s \in [0, T)} w_j y_{i, j, s} (s + p_{i,j})   \tag{$\lpinterval$} 		
		\label{LP:interval}
	\end{equation}\vspace*{-22pt}
$\qquad\qquad\qquad \text{s.t} $
	\begin{alignat}{2}
		\sum_{s \in [0, T)}y_{i, j, s} &= x_{i, j} &\qquad \forall i &\in M, j \in J \label{lpi-1}\\
		\sum_{i \in M}x_{i, j} & = 1 &\qquad \forall j &\in J \label{lpi-2} \\
		\sum_{j \in J,\ s \in [\max\{0,t-p_{i,j}\}, t)} y_{i, j, s} &\leq 1 &\qquad \forall i &\in M, t \in [T] \label{LPC:cover-by-one-interval}\\
		y_{i, j, s} &\geq 0 &\qquad \forall i &\in M, j \in J, s \in [0, T)  \nonumber \\ %\label{lpi-4} \\
		y_{i, j, s} & = 0 &\qquad \forall i &\in M, j \in J, s < r_j \text{ or } s > T-p_{i,j} \nonumber \label{lpi-5} 
	\end{alignat}
	
 	To see this is a valid LP relaxation for non-preemptive schedules, assume that all variables can only take integer values. Then, the first two constraints require that each job $j$ must be assigned to exactly one machine, which is captured by the indicator variable $x_{i,j}$. The variable $y_{i,j,s} = 1$ if and only if $j$ starts getting processed at time $s$ on machine $i$. The constraints 
 (\ref{LPC:cover-by-one-interval}) ensure that only one job gets processed at a time on any machine. The last constraint prohibits jobs from getting processed before their arrival times. We obtain a valid LP relaxation by allowing variables to have fractional values.

 \paragraph{Rounding.}
	We now describe how to round the LP, which consists of two steps. The first step is to define a `pseudo' arrival time $\tau_j \geq r_j$ for each job $j$. For each job $j$, we can view $\{y_{i, j, s}\}_{i,s}$ as a probability distribution over pairs $(i,s)$ due to Constraint (\ref{lpi-1}), and choose a pair $(i_j, s_j)$ according to the distribution randomly and independently. Job $j$ will be scheduled on machine $i_j$. 
	Let $\bfTheta$ be some distribution over real numbers in $[0, 1]$ where no number in the distribution occurs with positive probability; $\bfTheta$ will be fixed later. We randomly and independently choose a number $\theta_j$ from $\bfTheta$. Define $\tau_j = s_j + \theta_j\cdot p_{i_j, j}$. We assume w.l.o.g. that all jobs have different $\tau_j$ values since this event happens almost surely. 
	
	In the second step, we finalize each machine's schedule. For each $i \in M$, let $J_i = \{j\in J:i_j = i \}$ be the set of jobs that are assigned to $i$.  Let $\pi$ be the ordering of $J_i$ according to increasing order of $\tau_j$ values.
	We schedule jobs in $J_i$ on machine $i$ according to $\pi$, pretending that  $\tau_j$ is job $j$'s actual arrival time. That is, job $j \in J_i $ starts when all jobs in $J_i$ ahead of $j$ in the ordering of $\pi$ complete, or at time $\tau_j$, whichever comes later. 

	Notice that if we use the actual arrival times $r_j$ instead of the pseudo ones for scheduling, we can obtain the optimum schedule on $i$ respecting the ordering $\pi$ -- that is, each job $j \in J_i$ starts when all jobs in $J_i$ before $j$ according to $\pi$ complete, or at time $r_j$, whichever comes later. The schedule given by our algorithm might be worse than this optimum schedule respecting $\pi$. However, for the sake of analysis, it is more convenient to use our schedule, rather than the optimum one. Our schedule on machine $i$ might have fractional starting times, but it is not an issue since we can convert the schedule to the optimum one respecting $\pi$, in which all starting times are integral. 
	
\subsection{Analysis}
	\label{sec:n-analysis}
	
		It will be convenient to think of the LP solution as a set $\calR_i$ of rectangles for each machine $i$. For each pair of $j$ and $s$ with $y_{i, j, s} > 0$, we have a rectangle $R_{i, j, s}$ of length $p_{i, j}$ and height $y_{i, j, s}$ in $\calR_i$.  Horizontally, the rectangle $R_{i, j, s}$ covers the time interval $(s, s+p_{i, j}]$. For any machine $i$, the total height of rectangles in $\calR_i$ covering any time point $t \in (0, T]$ is at most $1$.

	We will analyze the expected completion time of each job $j$ and upper bound it by the corresponding LP quantity, $\sum_{i,s} y_{i, j, s} (s + p_{i,j})$. Towards this end, henceforth we fix a job $j \in J$, the machine $i \in M$ job $j$ is assigned to, and a value of $\tau \in (0, T]$ job $j$ is given. We consider $\E[C_j| i_j = i, \tau_j = \tau]$, i.e, the expected completion time of $j$, conditioned on the event that $i_j = i$ and $\tau_j = \tau$. For notational convenience, we use $\widehat \E[\cdot]$ to denote $\E[\cdot|i_j =i, \tau_j =\tau]$ and $\widehat \Pr[\cdot]$ to denote $\Pr[\cdot|i_j =i, \tau_j =\tau]$. After bounding $\widehat \E[C_j]$ by $\tau$ and $p_{i,j}$, we will get the desired bound on $\E[C_j]$ by deconditioning.
	
	The key issue we have to handle when jobs have arrival times is that there can be idle times before job $j$ starts. Hence we have to consider not only the volume of jobs scheduled before job $j$, but also the total length of idle times. 
	
	\begin{definition}
		For a time point $t \in (0, T]$, we say that $t$ is \emph{idle}, if there are no jobs scheduled at time $t$ on machine $i$ in our schedule. Let $\idle(t)$ indicate whether the time point $t$ is idle or not. 		
	\end{definition}
	
 	With this definition, we are ready to formally break down $C_j$ into several quantities of different characteristics.

\begin{equation}
\displaystyle C_j = \sum_{j' \in J_i: \tau_{j'} < \tau_j} p_{i, j'} + \int_{0}^{\tau_j} \idle(t)\sfd t+ p_{i, j}.
		\label{obs:decompose-C-j}
\end{equation}

	The first term is the total length of jobs scheduled before $j$ on machine $i$ and the second is the total length of idle times before $\tau_j$. Notice that there are no idle points in $[\tau_j, C_j)$ since all jobs $j' \in J_i$ scheduled before $j$ have $\tau_{j'} < \tau_j$. 

\vspace{-2mm}	
\paragraph{Uniform Distribution and its Limitations.}
	Before we present a better than 2-approximation, we take a short detour to discuss how we recover a simple 2-approximation 
	by setting $\bfTheta$ to be the uniform distribution over $[0, 1]$.  To compute $\widehat\E[C_j]$, we first consider $\widehat\E\big[\sum_{j'\in J_i, \tau_{j'} < \tau}p_{i, j'}\big]$.  If some $j'  \in J_i$ has $\tau_{j'} < \tau  = \tau_j$, we say that the pair $(j', s_{j'})$ contributed $p_{i, j'}$ to the sum.   For each $j' \neq j$ and integer $s < \tau$, the expected contribution of the pair $(j', s)$ to the sum is $\widehat \Pr\big[i_{j'} = i, s_{j'} = s\big]\widehat \Pr\big[\tau_{j'} < \tau |i_{j'} = i, s_{j'}=s\big]p_{i,j'} = y_{i, j', s}\cdot \min\{1, (\tau-s)/p_{i,j'}\} \cdot p_{i, j'} = y_{i, j', s}\cdot \min\{p_{i, j'}, \tau-s\}$. This is exactly the area of the portion of the rectangle $R_{i, j', s}$ before time point $\tau$. 
	Summing up over all pairs $(j' \neq j, s < \tau)$, $\widehat\E\big[\sum_{j'\in J_i, \tau_{j'} < \tau}p_{i, j'}\big]$ is at most the total area of the portions of $\calR_i$ before $\tau$, which is at most $\tau$.  The total length of idle slots before $\tau$ is obviously at most $\tau$. Thus, $\widehat\E[C_j] \leq 2\tau + p_{i, j}$. Since $\E[\tau_j|i_j = i, s_j = s] = s+p_{i,j}/2$, we have $\E[C_j |i_j = i, s_j = s] \leq 2(s+p_{i,j}/2)+p_{i,j} = 2(s+p_{i,j})$.  Since $\Pr[i_j = i, s_j = s] = y_{i, j, s}$, we have that $\E[C_j] \leq 2\sum_{i, s}y_{i,j,s}(s+p_{i,j})$, which is exactly twice the contribution of $j$ to the $\lpinterval$ objective.  Thus, we obtain a 2-approximation for the problem. 
	
	\smallskip
	However, uniform distribution does not yield a better than 2-approximation. To see this, consider the following instance and LP solution. There are $1 / \eps +1$ machines indexed by $1, 2, ..., 1/ \eps +1$. There is one unit-sized job $j^*$ with arrival time $r$ and it is scheduled on each of machines $1, 2, ..., 1/ \eps$  by $\eps$ fraction during $(r, r+1]$; $j^*$ is not allowed to be scheduled on machine $1 / \eps +1$. There are $1 / \eps$ big jobs of sizes $p \gg r$ with arrival time 0, which are indexed by $j_1, j_2, ..., j_{1/ \eps} $. Each big job $j_k$ can be assigned to either machine $k$ or machine $1 / \eps+1$. The job $j_k$ starts on machine $k$ at time 0 by $1 - \eps$ fraction, and on machine $1 / \eps+1$ by $\eps$ fraction. For simplicity, say the unit-sized job has a unit weight and the big jobs have zero (or infinitesimally small) weights so that the objective is essentially dominated by the unit sized job $j^*$'s completion time. Clearly, $j^*$ has completion time $r+1$ in the LP solution. 
	
	We now show that the above rounding makes $j^*$'s completion time arbitrarily close to $2r$ in expectation. Fix the machine $j^*$ is assigned to by the above algorithm; w.l.o.g. assume that the machine is 1. With $1 - \eps$ probability, job $j_1$ is assigned to machine 1; under this event, $j_1$ has a smaller $\tau$ value than $j^*$ with probability $r / p$.  Hence $j^*$ starts at time $p$ with probability $(1-\eps)r/p$, otherwise at time $r$, meaning that $j^*$'s expected starting time is at least $(1 - \eps) (r / p) \times p + (1 - (1-\eps)r/p) \times r$ which tends to $2r$ as $\eps \rightarrow 0$ and $p \rightarrow \infty$. This shows one cannot get a better than 2-approximation using uniform distribution. 

\vspace{-2mm}	
\paragraph{Finding a Better Distribution.} 	
	The above example is simple yet illuminating. We first observe that pushing back the small job a lot due to big job might be a sub-optimal choice. Intuitively, a bigger job is less sensitive to delay since the delay can be charged to the job's processing time. We could try to shift mass in the distribution $\bfTheta$ to the right. Then, big jobs will be less likely to have smaller $\tau$ values than the small job. However, this could increase $\tau_j$ values in expectation, thereby increasing the objective. We would like to avoid increasing the offset added to $\tau_j$ which was $p_{i,j} /2$ (assuming that job $j$ goes to machine $i$).  To satisfy both requirements, we shall shift the mass from both ends to the middle. In the above example, the job $j^*$ overlaps the left-end of the big job $j_1$. Shifting the mass from the left to the middle will decrease the probability that $\tau_{j_1}<\tau_{j^*}$.  On the other hand, shifting the mass from the right to the middle will decrease the expectation of $\tau_{j^*}$. 

The remainder of this section is devoted to studying the effect of using different distributions on the approximation ratio.  Let $f:[0, 1]\to \R_{\geq 0}$ be the probability density function (PDF) of $\bfTheta$ and $F(t) = \int_{0}^t f(t')\sfd t'$ be the cumulative distribution function (CDF) of $\bfTheta$. 
Recall that we fixed a job $j \in J$, the machine $i \in M$ job $j$ is assigned to, and a value of $\tau \in [0, T)$ job $j$ is given. 
For every $j' \in J\setminus j, t \in (0, T]$ and integer $s$, we shall use $s \lhd_{j'} t$ to indicate that $s \in \left[\max\{0,t - p_{i, j'}\}, t\right)$. In other words, $s \lhd_{j'} t$ means that if $j'$ starts at $s$, then it must get processed at time $t$. 
For every $t \in (0, T]$,  define 
		\begin{align*}
			g(t) = \sum_{j' \in J \setminus j, \ s \lhd_{j'} t}y_{i, j', s}\cdot f\left(\frac{t-s}{p_{i, j'}}\right) \qquad \text{and} \qquad
			h(t) = \sum_{j' \in J \setminus j, \ s \lhd_{j'} t}y_{i, j', s} \cdot F\left(\frac{t-s}{p_{i, j'}}\right).
		\end{align*}
		
It is worth mentioning that $\sum_{s \lhd_{j'} t} y_{i, j', s}\cdot f\left(\frac{t-s}{p_{i, j'}}\right) \frac{1}{p_{i,j'}}$ is the density of the probability that $\tau_{j'}  = t$. Thus, integrating $g(t)$ from time 0 to $\tau = \tau_j$ will give the expected volume of work done before job $j$, which is the first term of (\ref{obs:decompose-C-j}) in expectation. The usefulness of $h(t)$ will be discussed shortly.

\begin{lemma}  \label{lemma:bound-length}
	$\displaystyle \widehat \E\left[\sum_{j' \in J_i: \tau_{j'}< \tau}p_{i, j'}\right] = \int_{0}^\tau g(t)\sfd t. $
\end{lemma}
\vspace*{-15pt}

	\begin{flalign*}
	Proof. &&
			\text{LHS} &= \sum_{j' \in J \setminus j} p_{i, j'}\cdot \widehat \Pr[i_{j'} = i, \tau_{j'} < \tau] 
		=  \sum_{j' \in J \setminus j} p_{i, j'}\sum_{s \in [0, \tau)}y_{i, j', s} \cdot \widehat \Pr\big[\tau_{j'} < \tau |i_{j'} = i, s_{j'} = s\big] \\
		&&	&=  \sum_{j' \in J \setminus j} p_{i, j'}\sum_{s \in [0, \tau)}y_{i, j', s}\cdot \int_0^{\min\{(\tau-s)/p_{i, j'}, 1\}} f(\theta) \sfd \theta \\
		&&	&=  \sum_{j' \in J \setminus j} p_{i, j'}\sum_{s \in [0, \tau)}y_{i, j', s}\cdot \frac{1}{p_{i,j'}}\int_s^{\min\{\tau, s + p_{i, j'}\}} f\left(\frac{t-s}{p_{i,j'}}\right) \sfd t \\
		&&	&= \int_{t=0}^\tau \sum_{j' \in J \setminus j}\sum_{s \lhd_{j'} t }y_{i, j', s} \cdot f\left(\frac{t-s}{p_{i,j'}}\right) \sfd t = \int_0^\tau g(t) \sfd t. \hspace*{135pt} \qed
	\end{flalign*}

We now shift our attention to bounding the second term in  (\ref{obs:decompose-C-j}) using the function $h(t)$. As we observed when using uniform distributions, the obvious upper bound on the second term is $\tau$. To improve upon this, we need to show a considerable fraction of times are not idle. We note that  
$\sum_{s \lhd_{j'} t}y_{i, j', s} \cdot F\left(\frac{t-s}{p_{i, j'}}\right)$ is the probability that job $j'$ is processed at time $t$ when starting at $\tau_{j'}$. If such an event occurs, then time $t$ will be shown to be non-idle, hence we get some credits.

\begin{claim} \label{claim:h-as-most-1}
	$h(t) \leq 1$ for every $t \in [0, T)$. 
\end{claim}
\begin{proof}
	Since $F$ is a CDF, we have $F(t') \leq 1$ for every $t' \in [0, 1]$. Thus, 
	$h(t) \leq \sum_{j' \in J \setminus j, \ s \lhd_{j'} t} y_{i, j', s} \leq 1$
by Constraint~\eqref{LPC:cover-by-one-interval}.
\end{proof}

\begin{lemma} \label{lemma:bound-idle-t}
	For every $t \in (0, \tau]$, we have
		$\displaystyle \widehat E[\idle(t)] \leq e^{-h(t)}.$
\end{lemma}
\begin{proof}
	We say $t'$ is \emph{empty} if there are no jobs $j' \in J_i$ such that $t' \in (\tau_{j'}, \tau_{j'} + p_{i, j'}]$; let $\emp(t')$ denote the indicator variable that is 1 iff $t'$ is empty. We first observe that if some $t' \in (0, T]$ is not empty, then $t'$ is not idle. This is because a job $j'$ such that $t' \in (\tau_{j'}, \tau_{j'} + p_{i, j'}]$ is not processed at time $t'$ only when other jobs are. Thus, 
	\begin{flalign*}
		 &\;\;\;\;\;\;\widehat E[\idle(t)] \leq \widehat E[\emp(t)] \\
		 &= \prod_{j' \in J \setminus j} \left(1-\widehat\Pr\big[i_{j'} = i, t \in (\tau_{j'}, \tau_{j'} + p_{i, j'} ]\big]\right)
		\leq \exp\Big(-\sum_{j' \in J \setminus j} \widehat\Pr\big[i_{j'} = i, t \in (\tau_{j'}, \tau_{j'} + p_{i, j'} ]\big]\Big)\\
		 &\leq \exp\left(-\hspace{-2.5mm}\sum_{j' \in J \setminus j} \widehat\Pr\left[i_{j'} = i, t \in (s_{j'}, s_{j'} + p_{i, j'} ], \theta_{j'} < \frac{t-s_{j'}}{p_{i, j'}} \right]\right)			
	= \exp\left(-\hspace{-2.5mm}\sum_{j' \in J \setminus j} \sum_{s \lhd_{j'} t} y_{i, j', s} \cdot F\left(\frac{t-s}{p_{i, j'}}\right) \right) \\
		 &= e^{-h(t)}. \hspace*{416pt} \qedhere
	\end{flalign*}		
\end{proof}

\begin{lemma} \label{lemma:bound-idle}
	$\displaystyle \int_0^\tau \widehat \E[\idle(t)] \sfd t \leq \tau - \left(1-\frac1e\right) \int_0^\tau h(t)\sfd t$.
\end{lemma}
\begin{proof}
	By Lemma~\ref{lemma:bound-idle-t}, we have $\int_0^\tau \widehat \E[\idle(t)] \sfd t \leq \int_0^\tau e^{-h(t)}\sfd t$. Notice that $h(t) \in [0, 1]$ for every $t \in [0, 1]$ by Claim~\ref{claim:h-as-most-1}. Thus by the convexity of the function $e^{-x}$, we have that $e^{-h(t)} \leq (1-h(t))e^0 + h(t) e^{-1} = 1 - (1-1/e)h(t)$. Taking the integral from $t = 0$ to $\tau$ gives the lemma. 
\end{proof}
	
	\begin{lemma} \label{lemma:final-approx-ratio}
		Let $\rho = \sup_{\phi \in (0, 1]}\left(F(\phi) - \left(1-\frac1e\right) \int_{0}^{\phi} F(\theta) \sfd \theta\right)/\phi$,  $\beta = \int_{0}^1 f(\theta)\theta\sfd \theta$ and $\alpha = 1 + \max\{\rho, (1+\rho)\beta\}$. Then our algorithm is an $\alpha$-approximation algorithm.
	\end{lemma}
	
To prove Lemma~\ref{lemma:final-approx-ratio}, we first upper bound $\widehat\E[C_j]$ in terms of $\tau = \tau_j$ and $p_{i,j}$, then obtain an upper bound on $\E[C_j]$ by deconditioning. 

\begin{lemma}
	\label{claim:n1}
	$\widehat\E[C_j]  \leq (1+\rho)\tau + p_{i,j}.$
\end{lemma}
\begin{proof}
		By applying the bounds in Lemmas~\ref{lemma:bound-length} and \ref{lemma:bound-idle} to Eq.\ (\ref{obs:decompose-C-j}),  we have
		\begin{align}
			&\quad \widehat\E[C_j] - \tau - p_{i,j} \leq \int_0^\tau g(t)\sfd t- \left(1-\frac1e\right) \int_0^\tau h(t) \sfd t   \nonumber \\
			&= \int_0^\tau \sum_{j' \neq j, s\lhd_{j'} t }y_{i, j', s}\left(f\left(\frac{t-s}{p_{i,j'}}\right)-\left(1-\frac1e\right)F\left(\frac{t-s}{p_{i,j'}}\right)\right)\sfd t \nonumber\\
			&=  \sum_{j' \neq j, s \in [0, \tau)}y_{i, j', s}\cdot\int_s^{\min\{\tau, s+p_{i,j'}\}} \left(f\left(\frac{t-s}{p_{i,j'}}\right)-\left(1-\frac1e\right)F\left(\frac{t-s}{p_{i,j'}}\right)\right) \sfd t \nonumber\\
			&=  \sum_{j' \neq j, s \in [0, \tau)}y_{i, j', s}\cdot p_{i,j'}\cdot \int_0^{\min\{(\tau-s)/p_{i,j'}, 1\}} \left(f\left(\theta\right)-\left(1-\frac1e\right)F\left(\theta\right)\right) \sfd \theta. \nonumber
		\end{align}
		By the definition of $\rho$ and that $\int_0^\phi f(\theta)\sfd \theta = F(\phi)$ for $\phi \in [0, 1]$, we have
		\begin{align*}
			\widehat\E[C_j] &\leq \sum_{j' \neq j, s \in [0, \tau)} y_{i, j', s} \cdot p_{i, j'} \cdot \rho \cdot \min \{(\tau - s)/p_{i, j}, 1\}  + \tau + p_{i,j}\\
			&= \rho \sum_{j' \neq j, s \in [0, \tau)} y_{i, j', s} \cdot \min\{\tau -s, p_{i, j'}\} + \tau + p_{i,j} 
			\leq \rho \tau + \tau + p_{i,j} = (1+\rho)\tau + p_{i,j},
		\end{align*}
		where the last inequality holds because the sum is the total area of the portions of rectangles in $\calR_i$ before time $\tau$.
\end{proof}

\begin{lemma}
	\label{claim:n2}
	$\E[C_j]  \leq \alpha \sum_{i \in M, s \in [0, T)}y_{i, j, s}(s + p_{i,j}).$
\end{lemma}
\begin{proof}
		Now, we consider all machines $i \in M$. Then $\E[C_j]$ equals to
		\begin{align*}
			&\quad \sum_{i \in M, s \in [0, T)}\Pr[i_j = i, s_j = s]\int_0^1 f(\theta)\E[C_j|i_j = i, s_j = s, \tau_j = s + \theta p_{i, j}]\sfd \theta\\
			 &\leq \sum_{i \in M, s \in [0, T)}y_{i,j,s}\int_0^1 f(\theta)\left((1+\rho)(s+\theta p_{i,j})+p_{i,j}\right) \sfd \theta  \hspace{135pt} \mbox{[By Lemma~\ref{claim:n1}]}\\
			 &=\sum_{i \in M, s \in [0, T)}y_{i, j, s} \left(\int_0^1f(\theta)((1+\rho)s+p_{i,j})\sfd \theta + \int_0^1f(\theta)(1+\rho)\theta p_{i,j}\sfd \theta \right)
		\end{align*}		
		\begin{align*}	 
			&=\sum_{i \in M, s \in [0, T)}y_{i, j, s}\left((1+\rho)s+p_{i,j}+(1+\rho)\beta p_{i,j}\right) \\
			&\leq \sum_{i \in M, s \in [0, T)}y_{i, j, s}\max\{1+\rho, 1+(1+\rho)\beta\} (s + p_{i,j}) 
			= \alpha \sum_{i \in M, s \in [0, T)}y_{i, j, s}(s + p_{i,j}).  \hspace*{65pt} \qedhere
		\end{align*}
\end{proof}

We are now ready to complete the proof of Lemma~\ref{lemma:final-approx-ratio}. Summing up $\E[w_j C_j]$ over all jobs $j \in J$, we have \vspace{-3mm}
		\begin{align*}
			\E\Big[\sum_{j \in J}w_jC_j\Big] \leq \alpha \sum_{i \in M, j \in J, s \in [0, T)}w_j y_{i,j,s}(s+p_{i,j}).
		\end{align*}
		Notice that the right-hand-side is exactly $\alpha$ times the cost of the LP solution.  Thus, our algorithm is an $\alpha$-approximation.

\medskip
	To complete the proof of Theorem~\ref{thm:non-preemptive}, we only need to find a distribution $\bfTheta$ whose $\alpha$ value is no greater than the approximation ratio claimed in the theorem. We note that we first used a factor revealing LP to find out the best distribution that minimizes $\alpha$. Then we discovered a truncated quadratic function is the best fit for the obtained discretized PDF. To find the best coefficients, we ran another program and obtained a distribution that yields a slightly better approximation ratio than one we could using the factor revealing LP. We set the PDF $f$ as follows:
	\begin{equation}
		f(\theta) = \begin{cases}
				0.1702 \theta^2 + 0.5768 \theta + 0.8746	&\mbox{if} \;0 \leq \theta \leq 0.85897 \\
				0 			&\text{otherwise}
			\end{cases}.
	\end{equation}
	
	Notice that $f(\theta)$ increases as $\theta$ goes from $0$ to $0.85897$ and becomes $0$ when $\theta > 0.85897$. This is consistent with the previous discussion that we shift the probability mass from both ends to the middle.  Then, by easy calculation one can show that $\beta  < 0.46767$ and $\rho <  0.8785$. Thus $\alpha = 1+\max\set{\rho, (1+\rho)\beta} < 1.8786$. Details on this calculation can be found in Appendix~\ref{sec:alpha-calc}.

	\newcommand{\cA}{\mathcal{A}}

\section{Preemptive Scheduling}
	\label{sec:preemptive}

	This section is devoted to proving Theorem~\ref{thm:preemptive}, which claims a better than 2-approximation for the preemptive case. Note that migration is not allowed, i.e. each job must be processed on only one of the machines. In the preemptive setting, a job's processing may be interrupted, so we need to choose $p_{i,j}$ unit-length time slots on machine $i$ to schedule job $j$ on machine $i$.  This motivates the following definition. 
	
	\begin{definition}[Chains] \label{def:chain}
	 A chain $A$ for job $j \in J$ on machine $i \in M$ is a sequence $(t_1, t_2, \cdots, t_{p_{i, j}})$ of integers such that $r_j < t_1 < t_2 < \cdots < t_{p_{i,j}} \leq T$.  Equivalently, we may view  $A$ as the set $\{t_1, t_2, \cdots, t_{p_{i,j}} \}$, or as a function from $(0, p_{i,j}]$ to $(0, T]$ such that $A(\vartheta) = t_{\ceil{\vartheta}} + \vartheta-\ceil{\vartheta}$ for all $\vartheta \in (0, p_{i, j}]$. For all $t \in (0, T]$, let $A^{-1}(t) = \sup\{\vartheta \in (0, p_{i,j}]:A(\vartheta) \leq t\}$.
	\end{definition}
	
	A chain $A = (t_1, t_2, \cdots, t_{p_{i, j}})$ completely describes $j$'s schedule on machine $i$: we schedule $j$ on slots $(t_1-1, t_1], (t_2-1, t_2], \cdots, (t_{p_{i,j}}-1, t_{p_{i,j}}]$. Thus, $A(\vartheta)$ is the time at which we have run $j$ for $\vartheta$ units of time. In particular, $A(p_{i,j})$ is the completion time of $j$. 
	We may use $C_A:= A(p_{i,j})$ to denote $j$'s completion time under the schedule $A$ of job $j$. 	Notice that  $A^{-1}(t)$ is the amount of time in which $j$ is processed before $t$ in $A$. Let $\calA^{i,j}$ denote the set of all chains for job $j$ on machine $i$.  

\vspace{-2mm}	
\paragraph{Linear Programming.} 	We are now ready to present our LP using  the notion of chains.  For notational convenience, when we refer to a chain $A$, we assume it is associated with a machine $i$ and a job $j$ satisfying $A \in \calA^{i, j}$.
	\begin{equation}
		\min \qquad \sum_{i \in M, j \in J} \sum_{A \in \cA^{i,j}} w_j C_A  \cdot z_A  \tag{$\lpchain$} 		
	\end{equation}\vspace{-22pt}
$\qquad\qquad\qquad \text{s.t} $	
	\begin{alignat}{2}
		\sum_{i \in M} \sum_{A \in \cA^{i,j}} z_A &\geq 1        &&\qquad \forall j \in J \label{lpc-1}\\
		\sum_{j \in J} \sum_{A \in \cA^{i,j} : t \in A} z_A &\leq 1        &&\qquad \forall i \in M, t \in [T] \label{lpc-2}\\
		z_A &\geq 0 								&&\qquad \forall i \in M, j \in J, A \in \cA^{i,j}  \nonumber
	\end{alignat}

	To see $\lpchain$ is a valid relaxation, assume that variables can only take integer values. In $\lpchain$ we have an indicator  variable $z_A$ for every possible chain $A \in \cA^{i,j}$ for all $i$ and $j$, which is 1 if and only if $j$ is scheduled following the chain description $A$. The first constraint requires that every job must complete; note that we do not need equality here since the optimal solution will satisfy equality. It is also worth mentioning that job $j$ never gets processed before its arrival time $r_j$ since $j$'s chains don't allow it. Finally, the second constraint ensures that every machine is used by at most one job at any point in time -- there is at most one chain that schedules a job at any time. Thus we get a valid LP relaxation by allowing variables to have fractional values. 
	
	Although the LP has exponentially many variables, we can solve it using standard techniques -- we solve the dual using the Ellipsoid method with a separation oracle. To keep the flow of presentation, details are deferred to Section~\ref{sec:solvelp}.
	
\vspace{-2mm}
\paragraph{Algorithm.}  Our rounding is a natural generalization of the rounding for non-preemptive scheduling. To see this, suppose that a chain $A \in \calA^{i,j}$ is a sequence of $p_{i, j}$ consecutive integers. Then $A$ corresponds to an interval. If every chain in the support of $z$ corresponds to an interval, then the fractional solution is a valid solution to $\lpinterval$ for non-preemptive scheduling. In this scenario, our rounding works exactly in the same way as that for non-preemptive scheduling. Thus, we can generalize the former rounding by generalizing intervals to chains. 
	
	More specifically,  our rounding algorithm works as follows.   Let $\bfTheta$ be some distribution over $[0, 1]$. For every $j \in J$, we randomly and independently choose a pair $(i_j, A_j)$ such that $\Pr[(i_j, A_j) = (i, A)] = z_A$ for every $i \in M, A \in \calA^{i,j}$.   As $\sum_{i \in M, A \in \calA^{i,j}}z_A = 1$ for every $j$, the random procedure is well-defined. For each $j$, we randomly and independently choose a number $\theta_j$ from $\bfTheta$. Let $\tau_j = A_j(\theta_j\cdot p_{i_j, j})$. We assume that all jobs have different $\tau_j$ values since the event happens almost surely.  As in the algorithm for the non-preemptive scheduling, we let $J_i = \{j\in J:i_j = i \}$ and schedule all jobs in $J_i$ on machine $i$ in increasing order of $\tau_j$.  We schedule the jobs as early as possible, maintaining the property that job $j$ starts no earlier than $\tau_j$.  Notice that the schedule our algorithm constructed is non-preemptive, even though the problem allows preemption.
	
\vspace{-2mm}
\paragraph{Overview of the Analysis.} The analysis is more involved than the one for the non-preemptive case. To see this, let's recall how we gave a better than 2-approximation for the non-preemptive case. We can still break down a job's completion time as in Eq.~\eqref{obs:decompose-C-j} where job $j$'s completion time is decomposed into three quantities: total volume of jobs with smaller $\tau$ values, total length of idle times before $\tau_j$, and the size of job $j$ itself. As we observed, if we use a uniform distribution for $\bfTheta$, it is easy to get a 2-approximation by showing that both quantities are bounded by $\tau_j$, which is $j$'s starting time plus half of its size in expectation. Then, by using a non-uniform distribution $\bfTheta$ with more mass around the center, we could have the following benefits: (i) if a job $j' \neq j$ is processed a little before $\tau_j$, it is less likely to have a smaller $\tau$ value; and (ii) otherwise, a considerable fraction of job $j'$ is processed before $\tau_j$, thus contributes to reducing the number of idle times. Then, using the non-preemptive structure of the schedule, we were able to 
analyze each time's contribution to the first and second quantities in Eq.~\eqref{obs:decompose-C-j}. 

While the high-level idea is the same, we have to take a different analysis route for the preemptive case since each job's schedule is scattered over time, which keeps us from defining $h$. Note that many jobs may contribute to making a time $t$ busy since we don't have a nice structural property given by the intervals but not by the chains.
In particular, when a lot of jobs are partially processed around time $t$, the time will highly likely to become non-idle. This create an issue for the analysis since we don't get enough idle times compared to the volume of jobs we used. 

Hence we have to bound $C_j$ by taking a more global view of the schedule. In the analysis, we will consider two cases. Let $W$ denote the volume of work done by LP before $\tau_j$. If $W$ mostly comes from jobs that are processed very little before $\tau_j$, we can reduce the first quantity in \eqref{obs:decompose-C-j} using the non-uniform distribution. Otherwise, we can show that a large fraction of $W$ comes from jobs that are processed a lot by the LP by time $(9/10) \tau_j$. Then, either a lot of jobs complete by time $\tau_j$ or the entire interval $[(9/10) \tau_j, \tau_j]$ becomes non-idle. In either case, we can have a better  bound on the second quantity in Eq.~\eqref{obs:decompose-C-j} than the trivial $\tau_j$. Somewhat subtle definitions are needed for the analysis, but this is a high-level overview.

\subsection{Solving the LP}
	\label{sec:solvelp}

\smallskip
We first assume that $T$ is polynomially bounded and discuss later how to handle large $T$. The dual of $\lpchain$ is as follows. 

\begin{equation}
		\max \qquad \sum_{j \in J} \eta_j - \sum_{i \in M, t} \xi_{i,t} \tag{$\lpdual$} 		
\end{equation}\vspace{-22pt}
$\qquad\qquad\qquad \text{s.t} $	
	\begin{alignat}{2}
		\eta_j - \sum_{t \in A}  \xi_{i,t} 	&\leq w_j C_A        &&\qquad \forall i,j,  A \in \cA^{i,j} \label{lpd-1} \\
		\eta_j 					&\geq 0 			&& \qquad \forall j \in J \nonumber \\
		\xi_{i,t} 						&\geq 0 			&& \qquad \forall i \in M, t \in [T] \nonumber
	\end{alignat}

Note that $\lpdual$ has polynomially many variables, but exponentially many constraints. To solve the dual, we use the Ellipsoid method. Fortunately, there is a very simple separation oracle. Fix $i$ and $j$, and $C \in [r_j+p_{i,j}, T]$. Our goal is to find $A$ with $C = C_A$ for which Constraint (\ref{lpd-1}) is violated if such $A$ exists. Since the right-hand-side $w_j C_A$ and $\eta_j$ are fixed, it suffices to find a chain $A \in \cA^{i,j}$ that minimizes $\sum_{t \in A}  \xi_{i,t}$ and completes job $j$ exactly at time $C$. Thus, 
we only need to consider the set $A$ consisting of $p_{i,j} - 1$ slots $(t-1, t]$ in $(r_j, C-1]$ with the smallest $\xi_{i,t}$ values, and the slot $(C-1,C]$. If Constraint (\ref{lpd-1}) is violated for this $A$, we found a violated constraint. Otherwise, all the  constraints are satisfied for the fixed $i, j$ and $C$.

Using the Ellipsoid method with the above separation oracle, we can obtain a basic optimal solution of $\lpdual$, in which the number of tight constraints is bounded by the number of variables. Since $\lpdual$ has polynomially many variables, the basic solution makes only polynomially many constraints tight. Due to the strong duality, there is an optimal solution to the $\lpchain$ where all variables corresponding to dual constraints that are not tight are zero. Hence we can obtain an optimal solution to $\lpchain$ with a poly-sized support. 

\medskip
	We now extend this argument to large $T$. Let $R := \{ r_j \; | \; j \in J\}$ be the set of all jobs' arrival times. We add to $R$, exponentially increasing time steps, i.e. $\lceil (1+\eps)^k \rceil$ for all integers $0 \leq k \leq K$ where $K$ is the smallest $k$ such that $\lceil (1+\eps)^k \rceil \geq T$. If $R = \{t_0, t_1, t_2, ... \}$, we break $(0, T]$ into intervals $I_1 = (t_0, t_1], I_2 = (t_1, t_2], \cdots$.	Note that jobs can arrive only at the beginning of the intervals. 	Also note that there are polynomially many intervals. 
	For each chain $A \in \calA^{i, j}$, $C_A$ is defined slightly differently from before: $C_A$ is $t_{k}$ where $k$ is the largest $k'$ such that $A$ takes some time slots from $I_{k'}$.   Notice that we have $A^{-1}(p_{i, j}) \leq  C_A \leq (1+\eps)A^{-1}(p_{i,j})$. That is, $C_A$ $(1+\epsilon)$-approximates the actual completion time of $j$ when it is scheduled following $A$. Constraint~(\ref{lpc-2}) is changed to 
	$$\sum_{j \in J} \sum_{A \in \cA^{i,j}} | I_k \cap A| z_A \leq t_k-t_{k-1}   \qquad \forall i \in M, k$$

The dual is changed to, 
\begin{equation*}
		\max \qquad \sum_{j \in J} \eta_j - \sum_{i \in M, k} (t_k - t_{k-1}) \xi_{i,k} 
\end{equation*}\vspace{-22pt}
$\qquad\qquad\qquad \text{s.t} $	
	\begin{alignat*}{2}
		\eta_j - \sum_{k}  |A \cap I_k| \xi_{i,k} 	&\leq w_j C_A        &&\qquad \forall i,j,  A \in \cA^{i,j} \label{lpd-1} 
	\end{alignat*}

The separation oracle is almost the same. The only change is that we are allowed to pick up to $t_k-t_{k-1}$ time slots from $I_k$ to find $A$ that is most likely to violate the constraint. All the remaining procedure is identical.

\subsection{Analysis}
	\label{sec:n-analysis}

	We fix $j \in J, i \in M$ and $\tau \in (0, T]$ and condition on the event that $i_j = i$ and $\tau_j = \tau$. Using the same notations as before, let $\widehat{\E}[\cdot]$ denote $\E[\cdot|i_j = i, \tau_j = \tau]$ and $\widehat{\Pr}[\cdot]$ denote $\Pr[\cdot|i_j = i, \tau_j = \tau]$.  For any $t \in (0, T]$, we say $t$ is idle if there are no jobs scheduled at time $t$ on machine $i$, and use $\idle(t)$ to indicate whether $t$ is idle or not. We will again use Eq.~\eqref{obs:decompose-C-j} for the analysis of $\E C_j$. 
	
	We do not try to optimize the approximation ratio. Rather we will use a distribution $\bfTheta$ that is very close to the uniform distribution to make the analysis more transparent. The probability density function (PDF) of $\bfTheta$ is $f(\theta)  = 1/(1-2\lambda)$ if $\theta \in (\lambda, 1-\lambda)$ and $f(\theta) = 0$ otherwise, where $\lambda \in (0, 0.005)$ is some constant to be decided later. Let $F(t) = \int_{0}^t f(\theta)\sfd \theta$ be its cumulative distribution function (CDF).  Note that this is a uniform distribution with small portion of both ends clipped out.  It is not hard to show that if $\lambda = 0$ then we can still obtain a 2-approximation. The following claims easily follow from elementary algebra. 
	
\begin{claim}
	\label{claim:npdf-1}
		For any $\theta, \theta'$ such that $\lambda \leq \theta \leq \theta' \leq 1$, we have $\frac{F(\theta)}{\theta} \leq \frac{\theta' - \lambda}{\theta' ( 1- 2 \lambda)}$.
\end{claim}
\begin{claim}
	\label{claim:npdf-2}
		For any $\theta, \theta'$ such that $\lambda \leq \theta' \leq 1/2$ and $\theta' \leq \theta \leq 1$, we have $\frac{F(\theta)}{\theta} \geq \frac{\theta' - \lambda}{\theta' ( 1- 2 \lambda)}$.
\end{claim}

	\newcommand{\sfl}{{\mathsf{l}}}
	\newcommand{\sfh}{{\mathsf{h}}}
	\newcommand{\sfg}{{\mathsf{g}}}
	
	We start by defining heavy and light chains. Roughly speaking, a chain $A \in \cA^{i,j}$ is said to be heavy if a considerable fraction of the corresponding job $j$ is processed before $\tau_j=\tau$, otherwise light. 
	
	\begin{definition}
		Given a chain $A \in \calA^{i, j'}$ for some job $j' \neq j$, we say $A$ is heavy if $A^{-1}(\tau) \geq p_{i, j'}/15$ and light otherwise. Let $\calA^{i, j'}_\sfh$ and $\calA^{i, j'}_\sfl$ be the sets of heavy and light chains in $\calA^{i, j'}$, respectively.  
	\end{definition}
	
	Let $W_\sfh = \sum_{j' \neq j, A \in \calA^{i, j'}_\sfh} z_A A^{-1}(\tau)$ and $W_\sfl = \sum_{j' \neq j, A \in \calA^{i, j'}_\sfl} z_A A^{-1}(\tau)$. Let $W = W_\sfh + W_\sfl = \sum_{j' \neq j, A \in \calA^{i, j'}} z_A A^{-1}(\tau)$.  Then, $W$ is the total area of the portions of the rectangle chains before $\tau$; here we view $z_A$ fraction of chain $A$ as a chain of rectangles with height $z_A$ on times in $A$. In the light of this view, we immediately have $W \leq \tau$. 

\smallskip
We continue our analysis by considering two cases depending on how much light/heavy chains contribute to $W$.

	\subsubsection{Case 1: $W_\sfl \geq W/3$.}  In this case, we focus on the expected total length of jobs scheduled on machine $i$ before $j$.  For a light chain, a large portion is after $\tau$.  Since the non-uniform distribution $\Theta$ moves the mass to the middle, it will give smaller expected total length if many chains are light. In the following, the first inequality is due to Claim~\ref{claim:npdf-1} with $\theta' = 1/ 15$ and $\theta = \frac{A^{-1}(\tau)}{p_{i,j'}}$.
	\begin{align*}
		&\quad \widehat\E\left[\sum_{j' \in J_i: \tau_{j'} < \tau}p_{i,j'}\right] = \sum_{j' \neq j, A \in \calA^{i,j'}}z_A F\left(\frac{A^{-1}(\tau)}{p_{i,j'}}\right)p_{i, j'} \\
		&\leq \frac{1}{1-2\lambda}\sum_{j' \neq j, A \in \calA^{i,j'}_\sfh}z_A \left(\frac{A^{-1}(\tau)}{p_{i,j'}}\right) p_{i,j'} +  \frac{(1/{15}-\lambda)}{(1/{15})(1-2\lambda)}\sum_{j' \neq j, A \in \calA^{i, j'}_\sfl}z_A  \left(\frac{A^{-1}(\tau)}{p_{i,j'}}\right) p_{i,j'}\\
		&= \frac{1}{1-2\lambda} \sum_{j' \neq j, A \in \calA^{i,j'}} z_A A^{-1}(\tau) - \left( \frac{1}{1-2\lambda}-\frac{1-{15}\lambda}{1-2\lambda}\right)\sum_{j' \neq j, A \in \calA^{i,j'}_\sfl} z_A A^{-1}(\tau) \\
		&= \frac{W}{1-2\lambda} - \frac{{15}\lambda}{1-2\lambda}W_\sfl \ \leq \ \frac{W}{1-2\lambda}-\frac{5\lambda W}{1-2\lambda}\ = \ \frac{1-5\lambda}{1-2\lambda} W \ \leq\ \frac{1-5\lambda}{1-2\lambda} \tau.
	\end{align*}
	
	Thus, we have 
	\begin{equation}
		\label{bound1}
			\widehat\E[C_j] \leq \left(2-\frac{3\lambda}{1-2\lambda}\right)\tau+p_{i,j}.
	\end{equation}

	\subsubsection{Case 2: $W_\sfh \geq 2W/3$.} In this case, we shall further divide heavy chains into good and bad ones. Roughly speaking, a good chain doesn't process the corresponding job too much very close to $\tau$. Intuitively, good chains will likely lead to the job being processed considerably before time $\tau$. We will show that there are `enough' good chains that will make a lot of times before $\tau$ non-idle.

	\begin{definition}
		We say a heavy chain $A \in \calA^{i, j'}$ for some $j' \neq j$ is good, if $A^{-1}(9\tau/10) \geq A^{-1}(\tau)/2$, and bad otherwise. Let $\calA^{i, j'}_{\sfh\sfg}$ and $\calA^{i, j'}_{\sfh\sfb}$ be the sets of good and bad heavy chains in $\calA^{i, j'}_\sfh$, respectively.  
	\end{definition}
		
	Let $W_{\sfh\sfg} = \sum_{j' \neq j, A \in \calA^{i, j'}_{\sfh\sfg}} z_A A^{-1}(\tau)$ and $W_{\sfh\sfb} = \sum_{j' \neq j, A \in \calA^{i, j'}_{\sfh\sfb}} z_A A^{-1}(\tau)$ respectively. So, $W_{\sfh} = W_{\sfh\sfg} + W_{\sfh\sfb}$. Next we show that there are not many bad chains. 
	
	\begin{claim}
		$W_{\sfh\sfb} \leq \tau/5$.
	\end{claim}
		\begin{proof}
		$\displaystyle W_{\sfh\sfb} = \sum_{j' \neq j, A \in \calA^{i, j'}_{\sfh\sfb}}z_AA^{-1}(\tau) \ \leq \ 2\sum_{j' \neq j, A \in \calA^{i, j'}_{\sfh\sfb}} z_A(A^{-1}(\tau) - A^{-1}(9\tau/10))\leq \tau/5.$
		\end{proof}

	Thus, we have secured lots of good chains, precisely $W_{\sfh\sfg} \geq 2W/3 - \tau/5$.
			
	\begin{lemma}
		\label{lem:910}
		$\displaystyle \int_{0}^\tau (1-\idle(t))\sfd t \geq  \min\left\{\tau/10, \sum_{j' \in J_i: \tau_{j'} < 9\tau/10} p_{i, j'} \right\}$.
	\end{lemma}
	\begin{proof}
		The LHS of the inequality is the total length of non-idle times on machine $i$ before time $\tau$. If all jobs in $J'= \{j' \in J_i: \tau_{j'} < 9\tau/10\}$ completes before $\tau$, LHS is at least $\sum_{j' \in J'}p_{i, j'}$. Otherwise, $t$ is not idle for any $t \in (9\tau/10, \tau]$ and thus LHS is at least $\tau/10$. 
	\end{proof}
	
	We will lower bound the expected value of the RHS in Lemma~\ref{lem:910} as follows.  Note that we only use jobs $j'$ that have good chains since other jobs are not very useful for deriving a lower bound. 
	The proof is somewhat technical, so we first derive an upper bound on $\widehat\E C_j$ assuming that the bound is true. 
	
	\begin{lemma}
		\label{lem:q}
			 $Q: = \widehat \E\min\left\{\frac{\tau}{10}, \sum_{j' \in J_i: A_{j'} \in \calA^{i,j'}_{\sfh\sfg}, \tau_{j'}<9\tau/10}p_{i, j'} \right\} \geq \frac{\gamma W_{\sfh\sfg}}{10}$ where $\gamma = \left(1-\frac1e\right)\frac{(1-30\lambda)}{30(1-2\lambda)}$.
	\end{lemma}	

	Lemmas~\ref{lem:910} and~\ref{lem:q} will give us an upper bound on the length of idle times before $\tau$. To bound the total expected volume of jobs with smaller $\tau$ values than job $j$, we use the following obvious bound. 
	\begin{align*}
		\widehat\E\left[\sum_{j' \in J_i: \tau_{j'} < \tau}p_{i,j'}\right] = \sum_{j' \neq j, A \in \calA^{i,j'}}z_A F\left(\frac{A^{-1}(\tau)}{p_{i,j'}}\right)p_{i, j'} \leq \frac{1}{1-2\lambda}\sum_{j' \neq j, A \in \calA^{i,j'}}z_A \left(\frac{A^{-1}(\tau)}{p_{i,j'}}\right) p_{i,j'}  = \frac{1}{1- 2\lambda} W.
	\end{align*}

	Applying these two bound to Eq.~\eqref{obs:decompose-C-j} and using the fact that $W \leq \tau$, we have 
	\begin{align}
		\hat\E[C_j] &\leq \frac{1}{1-2\lambda}W + \left(\tau - \frac{\gamma W_{\sfh\sfg}}{10} \right)+ p_{i, j}
		\leq  \frac{1}{1-2\lambda}W + \tau - \frac{\gamma}{10}\left(\frac{2W}{3}-\frac{\tau}{5}\right) + p_{i,j} \nonumber\\
		&\leq \frac{1}{1-2\lambda}\tau + \tau - \frac{\gamma}{10}\cdot\frac{7\tau}{15} + p_{i,j}
		 \ =\ \frac{2-2\lambda-(1-1/e)(1/30-\lambda)(7/150) }{1-2\lambda}\tau + p_{i,j}. \label{bound2}
	\end{align}

	This bound (\ref{bound2}) will be combined with (\ref{bound1}) for Case 1 in the following section to complete the analysis.

\bigskip
	The remainder of this section is devoted to proving Lemma~\ref{lem:q}. The main difficulty in lower bounding $Q$ is no matter how big the second term in $Q$ is, the quantity is capped at $\tau / 10$. Hence if jobs are very large compared to the cap, $Q$ can be very small. Fortunately, we have found lots of good chains. Good chains process their corresponding jobs considerably before $\tau$. This implies that such jobs cannot be very large compared to $\tau$. 

	For formal proof, we define a random function $\Psi(\alpha, p')$ over a vector $\alpha \in [0, 1]^{J \setminus j}$ and $p' \in \R_{\geq 0}^{J \setminus j}$ as follows. Initially let $S \gets 0$. Then for every $j' \neq j$, with probability  $\alpha_{j'}$, we let $S = S + p_{j'}$. Then let $\Psi(\alpha, p') = \min\{\tau /10, S\}$.  We define $\alpha^*_{j'} = \widehat{\Pr}[j' \in J_i, A_{j'} \in \calA^{i, j'}_{\sfh\sfg}, \tau_{j'} < 9\tau/10]$ and $p^*_{j'} = p_{i, j'}$ for every $j' \neq j$. Then $Q$ is exactly $\E\Psi(\alpha^*, p^*)$. The following lemma will allow us to increase job sizes while keeping their expected contribution to $S$ the same. 
	
	\begin{lemma}
		\label{lem:scale}
		If for some job $j' \neq j$ and some real number $a \geq 1$, we update $\alpha_{j'}$ to $\alpha_{j'}/a$ and $p'_{j'}$ to $ap'_{j'}$, then $\E\Psi(\alpha, p')$ can only decrease. 
	\end{lemma}
	\begin{proof}
		Let $S'$ be the contribution of $J \setminus \{j, j'\}$ to $S$ in the definition of $\Psi(\alpha, p')$. In other words, we fix the random events for all jobs except $j$ and $j'$. 	Then, $S = S' + p'_{j'}$ with probability $\alpha_{j'}$ and $S = S'$ with probability $1-\alpha_{j'}$. Assume that $S' \leq \tau/10$ since otherwise, $\Psi$ is the same for both cases. 
		
	Before updating $\alpha$ and $p'$, we have 
		\begin{align*}
			\E[\Psi|S']  &= \alpha_{j'} \min\{S' + p'_{j'}, \tau/10\} + (1 - \alpha_{j'}) \min\{S', \tau/10\} =  S' + \alpha_{j'}\min\{p'_{j'}, \tau/10-S'\} \\
			 &= S' + \min\{\alpha_{j'}p_{j'}, \alpha_{j'}(\tau/10-S')\}.
		\end{align*}
		
	After the update, $\alpha_{j'}$ becomes $\alpha_{j'}/a$ and $p'_j$ becomes $ap'_j$. Thus, after update,
		\begin{align*}
		\E[\Psi|S'] 			&= S' + \min\{\alpha_{j'}p_{j'}, (\alpha_{j'}/a)(\tau/10-S')\}, 
		\end{align*}	
	which shows that $\E[\Psi]$ can only decrease after the update. 
	\end{proof}
	
	The next step is to show that a large fraction of the second quantity in $Q$ (or $\sum_{j' \neq j}\alpha^*_{j'}p^*_{j'}$) comes from good chains. 	
	
	\begin{lemma}
		\label{lem:20}
		$\sum_{j' \neq j}\alpha^*_{j'}p^*_{j'} \geq \frac{1-30\lambda}{2(1-2\lambda)}W_{\sfh\sfg}$.	
	\end{lemma}
	\begin{proof}
	\begin{align*}
		\sum_{j' \neq j}\alpha^*_{j'}p^*_{j'}&=\sum_{j' \neq j, A \in \calA^{i, j'}_{\sfh\sfg}}z_A F\left(\frac{A^{-1}(9\tau/10)}{p_{i,j'}}\right)p_{i,j'} \\
		&\geq \frac{1/{30}-\lambda}{(1-2\lambda)/{30}}\sum_{j' \neq j, A \in \calA^{i, j'}_{\sfh\sfg}}z_A\left(\frac{A^{-1}(9\tau/10)}{p_{i,j'}}\right)p_{i,j'}
		\ =\ \frac{1-30\lambda}{1-2\lambda}\sum_{j' \neq j, A \in \calA^{i, j'}_{\sfh\sfg}}z_AA^{-1}(9\tau/10)\\
		&\geq \frac{1-30\lambda}{1-2\lambda}\cdot\frac12\sum_{j' \neq j, A \in \calA^{i, j'}_{\sfh\sfg}}z_AA^{-1}(\tau)\ =\ \frac{1-30\lambda}{2(1-2\lambda)}W_{\sfh\sfg}.
	\end{align*}
	The first inequality follows by observing that for any good chain $A$, we have $A(9\tau / 10)/p_{i,j'} \geq 1/30$ and applying Claim~\ref{claim:npdf-2} with $\theta' = 1/30$ and $\theta = \frac{A^{-1}(9\tau/10)}{p_{i,j'}}$.
	\end{proof}	
		
	Notice if $\alpha^*_{j'} > 0$ then $\calA^{i, j}_{\sfh\sfb} \neq \emptyset$. Thus, taking an arbitrary $A \in \calA^{i, j}_{\sfh\sfb}$, we have $p^*_{j'} = p_{i, j'} \leq 15 A^{-1}(\tau) \leq 15\tau$.  Initially let $\alpha = \alpha^*$ and $p' = p^*$. 
	Then we apply Lemma~\ref{lem:scale}: for every $j'$ such that $\alpha_{j'} > 0$, we scale down $\alpha_{j'}$ and scale up $p'_{j'}$ by the same factor so that $p'_{j'}$ becomes $15\tau$. After the update, we have $\E\Psi(\alpha,p')\leq Q$. Moreover, $\sum_{j' \neq j}\alpha_{j'} p'_{j'} \geq \frac{1-30\lambda}{2(1-2\lambda)}W_{\sfh\sfg}$ as the operations maintained the left-hand-side in the bound of Lemma~\ref{lem:20}.	Thus, $\sum_{j' \neq j}\alpha_{j'} \geq \frac{1-30\lambda}{30(1-2\lambda)}\frac{W_{\sfh\sfg}}{\tau}$. Now, consider the process for computing $\Psi(\alpha, p')$. The probability that we add some $p'_{j'} = 15\tau$ to $S$ is 
	\begin{align*}
		1-\prod_{j' \neq j}(1-\alpha_{j'}) &\geq 1-\prod_{j' \neq j}e^{-\alpha_{j'}} = 1-\exp\left(-\sum_{j' \neq j}\alpha_{j'}\right) \geq 1-\exp\left(-\frac{(1-30\lambda)W_{\sfh\sfg}}{30(1-2\lambda)\tau}\right)\\
		&\geq \left(1-\frac1e\right)\frac{(1-30\lambda)W_{\sfh\sfg}}{30(1-2\lambda)\tau} = \gamma \frac{W_{\sfh\sfg}}{\tau}.
	\end{align*}
	
	Thus, $\E\Psi(\alpha, p') \geq \frac{\gamma W_{\sfh\sfg}}{\tau}\cdot \frac{\tau}{10} = \frac{\gamma W_{\sfh\sfg}}{10}$. Note that the expectation is lower bounded by the probability multiplied by $\tau / 10$ since the total size of $S$ is capped at $\tau / 10$ in $\Psi$. This completes the proof of Lemma~\ref{lem:q}.

\subsubsection{Wrapping up: Combining the Two Cases}
	
		We set $\lambda = 1/5100$, then in both cases (Eq. (\ref{bound1}) and (\ref{bound2})), we have $\widehat \E[C_j] < 1.99942\tau + p_{i, j}$. For a chain $A \in \calA^{i, j}$, $\E[\tau_j|i_j = i, A_j = A]$ is at most $A(p_{i,j}) - p_{i, j}/2$; this is where $\lpchain$ with each chain's cost associated with the corresponding job's completion time plays a crucial role. Thus,
	\begin{align*}
		\E[C_j] &< \sum_{i \in M, A \in \calA^{i, j}}z_A \left(1.99942(C_A-p_{i,j}/2)+p_{i,j}\right) \\
		&= \sum_{i \in M, A \in \calA^{i, j}} z_A (1.99942 C_A + 0.00029p_{i,j}) \leq 1.99971\sum_{i \in M, A \in \calA^{i, j}}z_AC_A.
	\end{align*}
	This is exactly 1.99971 times the (unweighted) contribution of $j$ to the LP solution. Thus, our algorithm is a $1.99971$-approximation, implying Theorem~\ref{thm:preemptive}.

\appendix

	\section{The Strength of the Configuration LP}

\label{sec:configuration-LP}

In this section, we consider \emph{all} convex programmings of the following form:
\begin{equation}
	\min \qquad \sum_{i \in M}f_i(x_i) \qquad \text{s.t} \label{LPForm}
\end{equation}
\vspace*{-30pt}

\begin{alignat}{2}
	\sum_{i \in M} x_{i, j} &= 1 &\qquad \forall j &\in J; \label{LPC:job-scheduled-config} \\
	x_{i, j} &\geq 0 &\qquad \forall i &\in M, j \in J. \nonumber
\end{alignat}

In the above, $x_i = \big(x_{i,j}\big)_{j \in J} \in [0, 1]^J$ is the vector of $x_{i,j}$s for all jobs; and  $f_i$ is a convex function over $[0, 1]^J$. In order for this convex programming to be valid, we require:
\begin{enumerate}[label=(*)]
	\item For every $x_i \in \{0, 1\}^J$, $f_i(x_i)$ is at most the total weighted completion time of the optimal schedule of jobs $\{j:x_{i, j} = 1\}$ on machine $i$.
\end{enumerate}

In this form, the only connection across different machines is made by  Constraint~\eqref{LPC:job-scheduled-config}.  Except for this constraint, the machines are treated separately.  All known programmings for $R||\sum_j w_j C_j$ and for $R|r_j|\sum_j w_j C_j$, including our interval LP, the LPs of  \cite{sethuraman1999optimal, schulz2002scheduling}, the convex programming of \cite{Skutella01},  the SDPs of \cite{Skutella01, BansalSS15} and the configuration LP of \cite{sviridenko2013approximating}, are of this form. 

In the configuration LP of \cite{sviridenko2013approximating}, the definition of $f_i(x_i)$ is the minimum of $\sum_{S \subseteq J}y_{i,S}W_{i, S}$ over all vectors $\left(y_{i,S}\right)_{S \subseteq J}$ satisfying
\begin{alignat}{2}
	\sum_{S \subseteq J} y_{i, S} &= 1;  \nonumber \\
	\sum_{S \subseteq J: j \in S} y_{i, S} &= x_{i,j}, &\qquad \forall j &\in J; \label{LPC:configure-agrees-with-x}\\
	y_{i, S} &\geq 0; &\qquad \forall S &\subseteq J; \nonumber
\end{alignat}
where $W_{i, S}$ is the total weighted completion time of scheduling $S$ on machine $i$ optimally.\footnote{In the configuration LP of \cite{sviridenko2013approximating}, each $S$ is not just a set of a jobs, but an actual schedule of some jobs on $i$. However, it is easy to see that their version is equivalent to ours.}  

Let $f^*_i$ be the definition of the function $f_i$ in the configuration LP.  We shall show that for every convex function $f_i$ satisfying (*), we have $f_i(x_i) \leq f^*_i(x_i)$ for every $x_i \in [0, 1]^J$.   Indeed, consider the vector $\left(y_{i, S}\right)_{S \subseteq J}$ that defines the $f^*_i(x_i)$ value. We have
\begin{align*}
	f_i(x_i) =  f_i\left( \sum_{S \subseteq J}y_{i, S} v^S \right) \leq  \sum_{S \subseteq J}y_{i, S} f_i(v^S)  \leq \sum_{S \subseteq J}y_{i, S}W_{i, S} = f^*_i(x_i), 
\end{align*}
where $v^S$ is the indicator vector for $S$: $v^S_j = 1$ if $j \in S$ and $v^S_j = 0$ if $j \notin S$.  The first equality is due to  Constraint~\eqref{LPC:configure-agrees-with-x}.  The first inequality follows from the convexity of $f_i$ and the second from (*).  Thus, the configuration LP gives the largest possible $f_i(x_i)$ value for every $x_i \in [0, 1]^J$. 
	\section{Lower Bound}
	\label{sec:lb}
	
	In this section, we show a lower bound of $e/(e-1) - \epsilon \geq 1.581$ for any algorithm based on independent rounding; 
	see the discussion before Theorem~\ref{thm:lb} for the description.  We remark that our instance is very similar to the instance of \cite{goemans2002single} which gives an $e/(e-1)$-lower bound on the integrality gap of some time-indexed LP for $1|r_j|\sum_{j}w_jC_j$. 
	
	In our lower bound instance, the fractional solution is a convex combination of optimum integral solutions; however, the independent rounding algorithm gives a solution whose cost is at least $e/(e-1)-\epsilon$ times that of an optimum integral solution.  Thus the limitation of the independent rounding algorithm is irrespective of the underlying convex programming: even if the convex programming exactly captures the convex hull of all integral solutions, the rounding algorithm still produces a sub-optimum solution. 
	
	\newcommand{\sfs}{{\mathsf{s}}}
	
	We assume $1/\epsilon$ is an integer and $T$ is an integer multiple of $1/\epsilon^3$.  Assume that $T$ is sufficiently large.  The lower bound instance consists of $1 / \eps + 1$ \emph{identical} machines indexed by $1, 2, 3, \cdots, 1/\eps + 1$,  $1/\eps$ big jobs $j^\sfb_1, j^\sfb_2, \cdots, j^\sfb_{1/\eps}$ of size $T$ and $T$ small jobs $j^\sfs_1, j^\sfs_2, \cdots, j^\sfs_T$ of size $1$. Big jobs arrive at time $0$; the small job $j^\sfs_t$ arrives at time $t-1$ for every $t \in [T]$.  Each big job has weight $\frac{\eps}{e}$, and the small job $j^\sfs_{t}$ has weight $\frac{e^{-t/T}}{T}$.
	
	We now define the fractional solution, in which all big jobs are scheduled in the interval $(0, T]$ and the small job $j^\sfs_t$ is scheduled in the interval $(t-1,t]$ for every $t \in [T]$.  For every $i \in [1/\eps]$, $j^\sfb_i$ is scheduled on machine $i$ with fraction $1-\eps$, and on machine $1/\eps + 1$ with fraction $\eps$.  Every small job $j^\sfs_t$ is scheduled on every machine $i$ with fraction $\eps$.  Notice that this fractional solution is a convex combination of integral $1/\eps$ solutions, each with fraction $\epsilon$ in the combination. In the $i$-th integral solution, we schedule all small jobs on machine $i$, $j^\sfb_i$ on machine $1/\eps + 1$ and $j^\sfb_{i'}$ on machine $i'$ for every $i' \in [1/\eps]\setminus \{i\}$.  The cost of the fractional solution of \emph{any} valid convex programming is at most the cost of each integral solution, which is
	\begin{align}
		&\quad \sum_{t = 1}^T \frac{e^{-t/T}}{T}t  + \frac1\eps \cdot \frac{\epsilon}{e} \cdot T \ \leq \  \frac{1}{T}\int_{0}^{T}e^{-t/T}(t+1) \sfd t + \frac T e \nonumber \\
		&\leq T \int_{0}^{1}e^{-\theta}(\theta+1/T) \sfd \theta + \frac Te
		\ = \  -T (\theta+1)e^{-\theta}\big|_{\theta = 0}^1 + \left(1-\frac1e\right) +  \frac Te \nonumber \\
		&= \left(1-\frac2e\right) T +   \frac Te + \left(1-\frac1e\right) 
		\ = \  \left(1-\frac1e\right) (T+1). \label{fractional-cost}
	\end{align}
	\smallskip
	
	We now proceed to consider the expected cost of the solution produced by the independent rounding algorithm. We will only lower bound the expected cost on machine $1$ since machines $1, 2, 3, \cdots, 1/\epsilon$ are symmetric; we shall ignore the cost on machine $1/\eps + 1$. Since we only consider machine 1, we use $j^\sfb$ to denote $j^\sfb_1$, and $j_t$ to denote $j^\sfs_t$ for every $t \in [T]$.

	We observe that for any sufficiently large interval, almost $\eps$ fraction of jobs arriving during the interval are assigned to the machine with a high probability.

	\begin{lemma}
		\label{lem:uniform}
		With probability at least $1-\epsilon$, for every $0 \leq \ell \leq T - \epsilon T$, there are at least $(1-\epsilon)\epsilon^2 T$ values  $\ell' \in [\ell, \ell + \epsilon T)$, such that $j_{\ell'}$ is assigned to the machine. 
	\end{lemma}
	\begin{proof}
		This follows from standard concentration inequalities together with the fact that $T$ is sufficiently large.
	\end{proof}
	
	Thus, we proceed with our analysis assuming that the event in the lemma happens and $j^\sfb$ is assigned to machine 1. It is convenient to pretend that all small jobs are assigned to machine $1$, but their weights are scaled down by a factor of $\eps$.

	The only flexibility we have is to select an integer $\tau \in [0, T]$ and schedule all jobs in $J_i$ using the order $j_1, j_2, \cdots, j_{\tau}, j^\sfb, j_{\tau+1}, j_{\tau+2}, \cdots, j_{T}$ since small jobs arriving later have smaller weights.  The cost incurred by small jobs in $\{j_1, j_2, \cdots, j_\tau\}$ is, 
	\begin{align}
	& \sum_{\ell\in [\tau]: j_\ell \in J_i } \frac{e^{-\ell/T}}{T} \ell  \ \geq \  \sum_{k = 0}^{\floor{\tau/(\epsilon T)}-1} (1-\epsilon)\epsilon^2 T \frac{e^{-k\epsilon T / T}}{T} (k\epsilon T) \nonumber \\
		&= (1-\epsilon)\epsilon^3 T\sum_{k=0}^{\floor{\tau/(\epsilon T)}-1} e^{-k\epsilon} k \ \geq \  (1-\epsilon)\eps^3 T \int_{1}^{\floor{\tau/(\epsilon T)}} e^{-t\epsilon}t \sfd t   \nonumber\\
		&= (1-\eps)\eps T \int_{\epsilon}^{\epsilon \floor{\tau/(\epsilon T)}} e^{-\theta} \theta\sfd \theta \quad \geq\quad (1-\eps)\eps T\int_{\eps}^{\tau/T - \eps} e^{-\theta}\theta \sfd \theta. \label{small-first}
	\end{align}
	
	The first inequality follows by considering small jobs arriving during each interval $[k \eps T, (k+1) \eps T)$ for $0 \leq k \leq \floor{\tau/(\epsilon T)}-1$. Each of the intervals has length $\eps T$ hence has at least $(1 - \eps) \eps^2 T$ jobs arriving during the interval due to Lemma~\ref{lem:uniform}. Every job  arriving during $[k \eps T, (k+1) \eps T)$ has completion time at least $k \eps T$ and weight at least $e^{-(k+1) \eps T / T} / T = e^{-k\eps-\eps}/T$.
	
	\smallskip
	Note that the big job cannot start before $\tau - \eps T$ since at least one small job arrives during $[\tau - \eps T, \tau)$ due to Lemma~\ref{lem:uniform} and we decided to schedule the big job after the small job. Hence no job from  $\{j_{\tau + 1}, j_{\tau + 2}, \cdots, j_T\}$ can complete before $\tau - \eps T + T$ if it is assigned to  machine $1$.	The cost incurred by small jobs in $\{j_{\tau + 1}, j_{\tau + 2}, \cdots, j_T\}$ is,
	\begin{align}
		&\quad \left(\sum_{\ell: \tau < \ell \leq T, j_\ell \in J_i } \frac{e^{-\ell/T}}{T}\right)(\tau - \epsilon T + T)  
		\ \geq\ \frac{\tau + (1-\eps) T}{T}\cdot\sum_{k=1}^{\floor{(T-\tau)/(\eps T)}} (1-\epsilon)\eps^2 T e^{-(\tau + k\epsilon T)/T} \nonumber \\
		& = (1-\epsilon)\epsilon^2(\tau + (1-\eps) T) e^{-\tau / T} \sum_{k=1}^{\floor{(T-\tau)/(\eps T)}} e^{-k\eps} \nonumber
		 \ \geq\  (1-\eps)^2\eps^2(\tau + T) e^{-\tau/T} \int_{1}^{{\floor{(T-\tau)/(\eps T)}}+1} e^{-t\epsilon}\sfd t \nonumber\\
		&=(1-\eps)^2\eps(\tau + T) e^{-\tau/T}\int_{\epsilon}^{\epsilon{\floor{(T-\tau)/(\eps T)}}+\epsilon} e^{-\theta}\sfd \theta
		\ \geq\ (1-\eps)^2\eps(\tau + T) e^{-\tau/T}\int_{\epsilon}^{1-\tau/T} e^{-\theta}\sfd \theta.  \label{small-second}
	\end{align}
	
	In the above, the first inequality follows by considering small jobs arriving during each interval $[\tau + (k-1) \eps T, \tau + k \eps T)$ for $1 \leq k \leq \floor{(T-\tau)/(\eps T)}$.
	
	\smallskip
	Let $\beta = \tau /T$. Then the cost incurred by small jobs is at least,
	\begin{align*}
		(\ref{small-first}) + (\ref{small-second}) &\geq (1-\eps)\eps T\int_{\eps}^{\tau/T - \eps} e^{-\theta}\theta \sfd \theta + (1-\eps)^2\eps(\tau + T) e^{-\tau/T}\int_{\epsilon}^{1-\tau/T} e^{-\theta}\sfd \theta \\
		&\geq (1-\epsilon)^2\epsilon T\left( \int_{\epsilon }^{\beta - \epsilon} e^{-\theta} \theta \sfd \theta + (1+\beta)e^{-\beta}\int_{\epsilon}^{1-\beta} e^{-\theta}\sfd \theta\right) \\
		&\geq (1-\epsilon)^2\epsilon T\left( \int_{0}^{\beta} e^{-\theta} \theta \sfd \theta + (1+\beta)e^{-\beta}\int_{0}^{1-\beta} e^{-\theta}\sfd \theta -  5\epsilon \right) \\
		& = (1-\eps)^2\eps T \left(-(\theta+1)e^{-\theta}\big|_0^\beta + (1+\beta)e^{-\beta}(1- e^{-(1-\beta)})  - 5\epsilon \right) \\
		& = (1-\eps)^2\eps T\left(1-(\beta + 1)e^{-\beta} + (1+\beta)(e^{-\beta} - e^{-1}) - 5\epsilon\right) \\
		& = (1-\eps)^2\eps T \left(1 - (1+\beta)e^{-1} - 5\epsilon\right).
	\end{align*}

	As discussed above, the completion time of $j^\sfb$ is at least $\tau + T  - \eps T = (1+\beta - \eps)T \geq (1-\eps)^2(1+\beta)T$ and the weight of $j_b$ is $\epsilon / e $. So, the total weighted completion time of all jobs is at least,
	$$(1-\eps)^2\eps T \left(1 - (1+\beta)e^{-1} - 5\epsilon\right) + (1-\eps)^2\epsilon T(1+\beta)e^{-1} = (1-\eps)^2\eps T(1-5\eps) \geq (1-7\eps)\eps T.$$
	 Thus, using the independent rounding algorithm, we can not get an approximation ratio better  than 
	$$\frac{(1/\eps) (1-7\eps)\eps T}{ (1-1/e) (T+1)} = \frac{e}{e-1} - O(\eps) > 1.581 - O(\eps),$$ 
	proving  Theorem~\ref{thm:lb}.
		
	\section{Solving \ref{LP:interval} When $T$ is Not Polynomially Bounded}
	\label{sec:non-preemptive-a}

In Section~\ref{sec:non-preemptive}, we proved Theorem~\ref{thm:non-preemptive} assuming that all parameters are polynomially bounded by $n$ and $m$. In this section, we show that we can remove this simplifying assumption with a loss of $(1+\eps)$ in the approximation ratio for any constant $0 <  \eps \leq 1/ 2$. 

To make our algorithm work for arbitrary instances, we only need to reduce the size of $\lpinterval$ since the rounding algorithm runs in time polynomial in the number of non-zero variables in the LP solution. Towards this end, we will 
restrict jobs starting times to a poly-sized set $\bar{S}$ of times. This will reduce the number of $y$ variables. Also we will enforce Constraints (\ref{LPC:cover-by-one-interval}) only for polynomially many times. 

Let $\delta  = \eps / 2n$. Let $\bar{S} = \{0, 1, 2, ..., \lceil 1 / \delta \rceil,   \lceil (1+\delta) / \delta \rceil, \lceil (1+\delta)^2 / \delta \rceil, ..., \lceil (1+\delta)^K / \delta \rceil\}$ where $K$ is the smallest integer such that $(1+\delta)^K / \delta \geq (1+\eps) T$. 
We enforce that every job starts only at a time in $\bar{S}$, i.e. $s \in \bar{S}$ in $\lpinterval$. Since $|\bar{S}| = O(\frac{n}{\eps} + \frac{1}{\delta} \log T)$, we will have only polynomially many $y$ variables. We keep Constraints (\ref{LPC:cover-by-one-interval}) only for all $i$ and $t$ such that $t-1 \in \bar{S}$. Let's call the resulting LP the new LP in contrast to $\lpinterval$, which we will also call the old LP.

\medskip
It now remains to show two things: (i) there exists a feasible solution to the new LP whose value is at most $(1+\eps)$ times the integral optimum; and (ii) the solution to the new LP is also a feasible solution to the old LP -- in other words, we need to make sure that the new LP solution satisfies Constraints (\ref{LPC:cover-by-one-interval}) for all times $t$ and machines $i$. Here we need to increase the maximum time step $T$ considered in the LP to $(1+\eps)T$, but this is a minor detail.  

\smallskip
We begin by showing (i). Fix an optimum solution and consider any fixed machine $i$. It suffices to show that we can transform the optimal schedule on machine $i$ without increasing each job's start time by a factor of more than $1+\eps$.  We will only shift jobs' start times, thus jobs assigned to machine $i$ will remain there. For notational convenience, we rename jobs scheduled on $i$ as $1, 2, 3, ..., n'$ in increasing order of their starting times. 

Let $s^*_j$ denote $j$'s start time in the optimal solution. We define $\Delta_j$ recursively. Consider $j$ in increasing order and define $\Delta_j$ be the smallest non-negative integer such that $\bar{s}_j := s^*_j + (\Delta_1 + \Delta_2 + ... + \Delta_j) \in \bar{S}$. We start processing each job $j$ at time $\bar{s}_j$. It is easy to see that this schedule is feasible by inductively showing that one can feasibly schedule all jobs by starting jobs $1, 2, ..., j$ at times $\bar{s}_1, \bar{s}_2, ... \bar{s}_j$, respectively, and then the remaining jobs $\Delta_1 + \Delta_2 + ... + \Delta_j$ time steps later than their respective start times in the original optimal solution. Notice that now all jobs start at times in $\bar{S}$.

We now show by induction on $j$ that this shifting process increases each job's start time by a factor of at most $1+\eps$. 

\begin{claim}
	\label{claim:db}
	For any positive integer $s$, let $\Delta$ be the smallest non-negative integer such that $s + \Delta \in \bar{S}$. Then, $s + \Delta \leq (1 + \delta) s$.
\end{claim}
\begin{proof}
	Suppose $s > 1 / \delta$ since otherwise the claim is immediate. Then, we have $\lceil (1+\delta)^k / \delta  \rceil  +1 \leq s \leq \lceil (1+\delta)^{k+1} / \delta \rceil$ for some integer $k \geq 0$. Then $\Delta \leq \lceil (1+\delta)^{k+1} / \delta \rceil - \lceil (1+\delta)^k / \delta \rceil - 1 \leq (1+\delta)^{k+1}/\delta - (1+\delta)^k/\delta = (1+\delta)^k \leq \delta s$.
\end{proof}

\begin{lemma}
	\label{lem:db}
	For all $j$, $\bar{s_j} \leq  (1+ \eps) s_j$.
\end{lemma}
\begin{proof}
	We show by induction on $j$ that $\bar{s_j} \leq (1 + \delta)^j s_j$. This claim is immediate when $j = 1$ by Claim~\ref{claim:db}. Suppose the claim holds for all jobs $1, 2, ..., j-1$. From definition of $\bar{s}_{j-1}$, we have $\Delta_1 + \Delta_2 + ... \Delta_{j-1} \leq ((1+\delta)^{j-1} -1) s_{j-1} \leq ((1+\delta)^{j-1} -1) s_j$. Hence we have $s_j + \Delta_1 + \Delta_2 + ... \Delta_{j-1} \leq (1+\delta)^{j-1}s_j$. By applying Claim~\ref{claim:db} with $s = s_j + \Delta_1 + \Delta_2 + ... \Delta_{j-1}$ and $\Delta = \Delta_j$, we have $\bar{s_j} \leq (1 + \delta)^j s_j$. Then, for all $j$, we have $\bar{s_j} \leq (1 + \delta)^j s_j \leq \exp(\delta j) s_j \leq  (1 + 2 \delta j) s_j \leq (1+\eps) s_j$ where the second inequality follows since $\exp(z) \leq 1+ 2z$ for $z \leq 1/ 2$. 	
\end{proof}

\smallskip
We now shift our focus to proving (ii). Fix any feasible solution to the new LP. Fix a machine $i$. For the sake of contradiction, suppose Constraint (\ref{LPC:cover-by-one-interval}) is violated for some time;  let $t^*$ be the earliest such time step. Note that $t^*-1 \notin \bar{S}$ since we kept Constraint (\ref{LPC:cover-by-one-interval}) for all times in $\bar{S}$. Let $s$ be the latest time step before $t^*$ in $\bar{S}$. Notice that $s$ must exist since $0 \in \bar{S}$. Since $t^* \notin \bar{S}$, we know that no jobs start during $(s, t^*]$. What this means is that every job processed at time $t^*$ is also processed at time $t'=s+1$ at the same rate. This is a contradiction since  the total height of jobs processed at time $t'$ is at most one due to Constraint (\ref{LPC:cover-by-one-interval}) for time $t'$.

\section{Details on Calculating $\alpha$ in Section~\ref{sec:non-preemptive}}
	\label{sec:alpha-calc}
	
	In this section, we discuss how we compute $\alpha$ in detail. Although we found the distribution via programming, we can verify the approximation guarantee purely analytically.  Let $a = 0.1702, b = 0.5768, c = 0.8746$ and $d = 0.85897$.  So $f(\theta) = a\theta^2 + b\theta + c$ if $\theta \in [0, d]$ and $f(\theta) = 0$ if $\theta \in (d, 1]$. 
 For this $f$, we have
	\begin{flalign*}
		&& F(1)&=\int_0^1 f(\theta) \sfd \theta = \frac{ad^3}{3} + \frac{bd^2}{2} + cd  \approx 1.00000125 > 1, &&\\
		&& \beta &= \int_0^1 f(\theta) \theta \sfd \theta =  \frac{ad^4}{4} + \frac{bd^3}{3} + \frac{cd^2}{2} < 0.46767.\\
		\text{We define} && \rho(\phi)&:=\left(F(\phi)-\left(1-\frac1e\right)\int_0^\phi F(\theta)\sfd\theta \right)\frac1\phi.
	\end{flalign*}
	$F(\phi)-(1-1/e)\int_0^\phi F(\theta)\sfd \theta$ decreases as  $\phi$ goes from $d$ to $1$ since $F(\phi)$ remains a constant.  Thus, when $\phi$ goes from $d$ to $1$, $\rho(\phi)$ will decrease as long as $\rho(\phi)$ remains positive.  So, to compute $\rho = \sup_{\phi \in [0, 1]} \rho(\phi)$, it suffices to consider $\phi \in [0, d]$.  For $\phi \in [0, d]$, we have 
	\begin{align*}
		\rho(\phi) &= \left(\frac{a\phi^3}{3} + \frac{b\phi^2}{2}+c\phi - \left(1-\frac1e\right)\left(\frac{a\phi^4}{12}+\frac{b\phi^3}{6}+\frac{c\phi^2}{2}\right)\right)\frac1\phi \\
		&=\frac{a\phi^2}{3} + \frac{b\phi}{2}+c - \left(1-\frac1e\right)\left(\frac{a\phi^3}{12}+\frac{b\phi^2}{6}+\frac{c\phi}{2}\right), \\
		\frac{\sfd \rho(\phi)}{\sfd \phi} &= \frac{2a\phi}{3} + \frac{b}{2} - \left(1-\frac1e\right)\left(\frac{a\phi^2}{4}+\frac{b\phi}{3}+\frac{c}{2}\right) = A_2 \phi^2 + A_1 \phi + A_0,
	\end{align*}
	where $A_2 = -(1-1/e)a/4, A_1 = 2a/3-(1-1/e)b/3$ and $A_0 = b/2-(1-1/e)c/2$.
	Since $A_2$ is negative, we only need to consider $\rho(\phi)$ for $\phi = 0$, $\phi = d$ and $\phi = \phi^* = \frac{A_1+\sqrt{A_1^2-4A_0A_2}}{-2A_2} \approx 0.5338653$.
	\begin{align*}
		\rho(0) &= c=0.8746, \quad \rho(d) = \frac{ad^2}{3} + \frac{bd}{2}+c - \left(1-\frac1e\right)\left(\frac{ad^3}{12}+\frac{bd^2}{6}+\frac{cd}{2}\right) < 0.8763, \\
		\rho(\phi^*) &= \frac{a{\phi^*}^2}{3} + \frac{b{\phi^*}}{2}+c - \left(1-\frac1e\right)\left(\frac{a{\phi^*}^3}{12}+\frac{b{\phi^*}^2}{6}+\frac{c{\phi^*}}{2}\right) \approx 0.8784782 < 0.8785.
	\end{align*}
	
	Thus $\rho=\sup_{\phi\in(0, 1]}\rho(\phi) <  0.8785$. Thus $\alpha = 1+\max\set{\rho, (1+\rho)\beta} < 1.8786$.

\bibliographystyle{plain}
\bibliography{unrelated}

\end{document}